\RequirePackage{fix-cm}
\documentclass[smallcondensed]{svjour3}     
\smartqed  
\usepackage{graphicx}
\usepackage{amsmath}
\usepackage{float}
\usepackage{amssymb}
\usepackage{bm}
\usepackage{cite}
\usepackage[colorlinks,linkcolor=blue,anchorcolor=blue,citecolor=blue]{hyperref}
\hypersetup{urlcolor=blue}
\begin{document}

\title{On pure MDS asymmetric entanglement-assisted quantum error-correcting codes
}


\author{Ziteng Huang         \and
        Weijun Fang          \and
        Fang-Wei Fu
}


\institute{Ziteng Huang\at
              Chern Institute of Mathematics and LPMC, Nankai University, Tianjin 300071, China\\
              \email{hzteng@mail.nankai.edu.cn}
           \and
           Weijun Fang (Corresponding Author)\at
              Shenzhen International Graduate School, Tsinghua University,\\
              and PCL Research Center of Networks and Communications, Peng Cheng Laboratory,\\
              Shenzhen 518055, P. R. China\\
              \email{nankaifwj@163.com}
           \and
           Fang-Wei Fu\at
              Chern Institute of Mathematics and LPMC, and Tianjin Key Laboratory of Network and Data Security Technology, Nankai University,\\
              Tianjin, 300071, P. R. China\\
              \email{fwfu@nankai.edu.cn}
}

\date{Received: date / Accepted: date}

\maketitle

\begin{abstract}
Recently, Galindo et al. introduced the concept of asymmetric entangle-ment-assisted quantum error-correcting codes (AEAQECCs) from Calderbank-Shor-Steane (CSS) construction. In general, it's  difficult to determine the required number of maximally entangled states of an AEAQECC, which is associated with the dimension of the intersection of the two corresponding linear codes. Two linear codes are said to be a linear $l$-intersection pair if their intersection has dimension $l$. In this paper, all possible linear $l$-intersection pairs of MDS codes are given. As an application, we give a complete characterization of pure MDS AEAQECCs for all possible parameters.
\end{abstract}

\section{Introduction}
\label{intro}
The theory of quantum error-correcting codes has developed rapidly after the works of Shor \cite{1} and Steane \cite{2,3}. Calderbank et al. \cite{4} gave systematic methods to construct quantum codes via classical self-orthogonal codes (or dual containing codes) over finite fields, called \emph{Calderbank-Shor-Steane} (CSS) construction. For overcoming the constraint of self-orthogonality, Brun et al. \cite{5} introduced the \emph{entanglement-assisted quantum error-correcting} codes (EAQECCs) by sharing entanglement between encoder and decoder. Recently, several classes of EAQECCs have been constructed \cite{6,7,8,9,10}.

In \cite{11}, Ioffe and M$\acute{e}$zart noticed that phase-shift errors happened more likely than qudit-flip errors. Therefore, considering EAQECCs in the asymmetric quantum channels is a valuable problem. Galindo et al. \cite{12} introduced the concept of \emph{asymmetric entanglement-assisted quantum error-correcting} codes (AEAQECCs) and gave the \emph{Gilbert-Varshamov bound} for AEAQECCs. Then they presented the explicit computation of the parameters of AEAQECCs via BCH codes. Liu et al. \cite{13} constructed three new families of AEAQECCs by means of Vandermonde matrices, extended GRS codes and cyclic codes.

The required number of maximally entangled states of an AEAQECC is determined by the dimension of the intersection of the two corresponding linear codes. Two linear codes are said to be a linear $l$-intersection pair if their intersection has dimension $l$. In \cite{14}, Guenda et al. constructed linear $l$-intersection pairs of MDS codes over $\mathbb{F}_{q}$ with length up to $q+1$ for most of the parameters.

In this paper, we firstly complement the results in \cite{14}. Specifically, we construct linear $l$-intersection pairs of two MDS codes with parameters $[n, k_{1}, n-k_{1}+1]_{q}$ and $[n, k_{2}, n-k_{2}+1]_{q}$, where $(n,k_{1},k_{2},l)=(q,l+1,l+1,l)$ for $ 0\leq l\leq q-2$ and $n=q+1$ with $1 \in \{ l, k_{1}-l, k_{2}-l \}$ for $k_{1},k_{2}\leq q$. Moreover, we construct all possible linear $l$-intersection pairs of MDS codes over $\mathbb{F}_{2^{m}}$ with length $n=2^{m}+2\geq6$. In summary, assuming the validity of the MDS Conjecture (Conjecture 1 in Section 3), all possible linear $l$-intersection pairs of MDS codes are given. As an application, we give a complete characterization of pure MDS AEAQECCs for all possible parameters. We list our main results as follows.

Let $q\geq3$ be a prime power and $n,k_{1}, k_{2}, l$ be non-negative integers. There exists a linear $l$-intersection
pair of two MDS codes with parameters $[n, k_{1},n-k_{1}+1]_{q}$ and $[n, k_{2}, n-k_{2}+1]_{q}$ if one of the following conditions holds (see Theorem 8):\\
(i) $n\leq q+1$, $k_{1},k_{2}\leq n-1$, $\max \{k_{1}+k_{2}-n,0\} \leq l \leq \min\{k_{1}, k_{2}\}$ (except $(n,k_{1},k_{2},l)\in \{(q+1,2,1,1),(q+1,1,2,1) \}$); \\
(ii) $q=2^{m}\geq 4$, $n=q+2$, $(k_{1},k_{2})\in \{(3,q-1), (q-1,3),(3,3)\}$, $0\leq l\leq3$;\\
(iii) $q=2^{m}\geq 4$, $n=q+2$, $(k_{1},k_{2})=(q-1,q-1)$, $q-4\leq l\leq q-1$.

Let $q\geq3$ be a prime power and $n,k_{1}, k_{2}, l$ be non-negative integers. There exists  a  pure MDS  $[[n,k_{2}-l,(k_{1}+1)/(n-k_{2}+1),k_{1}-l]]_{q}$ AEAQECC if one of the following conditions holds (see Theorem 9):\\
(i) $n\leq q+1$, $k_{1},k_{2}\leq n-1$, $\max \{k_{1}+k_{2}-n,0\} \leq l <\min\{k_{1}, k_{2}\}$; \\
(ii) $q=2^{m}\geq 4$, $n=q+2$, $(k_{1},k_{2})\in \{(3,q-1), (q-1,3),(3,3)\}$, $0\leq l\leq 2$;\\
(iii) $q=2^{m}\geq 4$, $n=q+2$, $(k_{1},k_{2})=(q-1,q-1)$, $q-4\leq l\leq q-2$.

The organization of this paper is presented as follows. In Section 2, we introduce some notions and results about linear codes, linear $l$-intersection pairs, quantum codes and AEAQECCs. In Section 3, all possible linear $l$-intersection pairs of MDS codes are given. In Section 4, we give a complete characterization of pure MDS AEAQECCs for all possible parameters. In Section 5, we conclude this paper.
\section{Preliminaries}
In this section, we introduce some notions and results about linear codes, linear $l$-intersection pairs, quantum codes and AEAQECCs.
\subsection{Linear codes and linear $l$-intersection pairs}
Let $q$ be a prime power and $\mathbb{F}_{q}$ be a finite field with $q$ elements. An $[n, k]_{q}$ linear code is a $k$-dimensional subspace of $\mathbb{F}_{q}^{n}$. For two vectors $\bm{a}, \bm{b}\in \mathbb{F}_{q}^{n}$, the (\emph{Hamming}) \emph{weight} $wt(\bm{a})$ of $\bm{a}$ is the number of nonzero components of $\bm{a}$ and the (\emph{Hamming}) \emph{distance} $d(\bm{a},\bm{b})$ between $\bm{a}$ and $\bm{b}$ is the number of positions at which the corresponding components are different, i.e., $d(\bm{a},\bm{b})=wt(\bm{a}-\bm{b})$. For a subset $\mathcal{A}$ of $\mathbb{F}_{q}^{n}$, the (Hamming) weight $wt(\mathcal{A})=\min\{wt(\bm{a}):\bm{a}\in \mathcal{A}\setminus \{\bm{0}\}\}$. The minimum Hamming distance $d(C)$ of a code $C$ is the minimum Hamming distance between any two distinct codewords. For a linear code $C$, we have $d(C)=wt(C)$. If the minimum Hamming distance $d$ of an $[n, k]_{q}$ linear code is known, we refer to the code as an $[n,k,d]_{q}$ linear code. One of the relations among these parameters is the \emph{Singleton bound}, which says that any $[n, k, d]_{q}$ linear code has to satisfy that
$$d \leq n - k +1.$$
An $[n, k, d]_{q}$ linear code is called a \emph{maximum distance separable} (MDS) code if it achieves the Singleton bound with equality. Let $A_{i}$ be the number of codewords of Hamming weight $i$ in a linear code $C$. The list $A_{i}$ for $0\leq i\leq n$ is called the \emph{weight distribution} of $C$. The weight distribution of an MDS code is given as follows.
\begin{theorem}\cite[pp. 262-265]{15}
Let $C$ be an $[n,k,d]_{q}$ MDS code where $d=n-k+1$. Then for $d\leq i\leq n$,
$$A_{i}={n\choose i}(q-1)\sum_{j=0}^{i-d}(-1)^{j} {i-1\choose j} q^{i-d-j}.$$
\end{theorem}

For $\bm{a}=(a_{1},a_{2},\ldots,a_{n}) \in \mathbb{F}_{q}^{n}$ and $\bm{b}=(b_{1},b_{2},\ldots,b_{n}) \in \mathbb{F}_{q}^{n}$, their (Euclidean) inner product is defined as
$\langle \bm{a},\bm{b}\rangle =\sum^{n}_{i=1} a_{i}b_{i}$. The (Euclidean) dual code of $C$ is defined as
$$ C^{\bot} = \{ \bm{a} \in \mathbb{F}_{q}^{n} : \langle \bm{a},\bm{b}\rangle = 0 , \textnormal{ for  any } \bm{b} \in C \}.$$
If $C \subseteq C^{\bot}$ ($C = C^{\bot}$), then $C$ is called a \emph{self-orthogonal} (\emph{self-dual}) code.

As an important class of MDS codes, the \emph{generalized Reed-Solomon} (GRS) code is the main tool in this paper.
Let $\mathcal{A}=\{a_{1},a_{2},\ldots,a_{n} \}$ be a subset of $\mathbb{F}_{q}$ with $n$ distinct elements and $\bm{v} =(v_{1},v_{2},\ldots,v_{n})$ where $v_{1},v_{2},\ldots,v_{n}$ are nonzero elements (not necessarily distinct) in $\mathbb{F}_{q}$. The GRS code associated to $\mathcal{A}$ and $\bm{v}$ is defined as
$$ GRS_{k}(\mathcal{A},\bm{v})=\{(v_{1}f(a_{1}),\ldots ,v_{n}f(a_{n})) :f(x) \in \mathbb{F}_{q}[x], \deg (f(x)) \leq k-1 \}.$$
The \emph{extended} GRS code associated to $\mathcal{A}$ and $\bm{v}$ is defined as
$$ GRS_{k}(\mathcal{A}\cup\infty,\bm{v})=\{(v_{1}f(a_{1}),\ldots ,v_{n}f(a_{n}),f_{k-1}) : f(x) \in \mathbb{F}_{q}[x], \deg (f(x))\leq k-1 \} $$
where $f_{k-1}$ stands for the coefficient of $x^{k-1}$ in $f(x)$. It is well-known that (extended) GRS codes are MDS codes and so are their dual codes.

Linear codes $C_{1}$ and $C_{2}$ over $\mathbb{F}_{q}$ with length $n$ are called a linear $l$-intersection pair if $\dim(C_{1} \cap C_{2} ) = l$. Using basic linear algebra, the lemma which gives the range of $l$ is given as follows.
\begin{lemma}\cite[Lemma 2.2]{14}
If there exists a linear $l$-intersection pair of  two linear codes with parameters $[n, k_{1}]_{q}$ and $[n, k_{2}]_{q}$, then $\max \{k_{1}+k_{2}-n,0\} \leq l \leq \min \{k_{1}, k_{2}\}$.
\end{lemma}
To determine the dimension $l$ of the intersection of two linear codes, the following lemma gives a relation between $l$, generator matrices and parity check matrices of the two corresponding codes, which is useful in our constructions in Section 3.
\begin{lemma}\cite[Theorem 2.1]{14}
Let $C_{1}$ be an $[n,k_{1}]_{q}$ linear code with generator matrix $G_{1}$ and  $C_{2}$ be an $[n,k_{2}]_{q}$ linear code with parity check matrix $H_{2}$. Then $\dim(C_{1}\cap C_{2})=l$ if and only if  $\textnormal{rank} (G_{1}H_{2}^{T})=k_{1}-l$.
\end{lemma}
\subsection{Quantum codes and AEAQECCs}
Let $\mathbb{C}$ be the complex field and $\mathbb{C}^{q}$ be the $q$-dimensional Hilbert space over $\mathbb{C}$. A quantum state is called a \emph{qubit} which is a nonzero vector of $\mathbb{C}^{q}$. A qubit $|v\rangle$ can be expressed as
$$|v\rangle=\sum_{a \in \mathbb{F}_{q}}v_{a}|a\rangle,$$
where $\{|a\rangle:a\in \mathbb{F}_{q}\}$ is a basis of $\mathbb{C}^{q}$ and $v_{a} \in \mathbb{C}$. An $n$-qubit is a nonzero vector in the $q^{n}$-dimensional Hilbert space $(\mathbb{C}^{q})^{\otimes n}\cong \mathbb{C}^{q^{n}}$, which can be expressed as
$$|\bm{v}\rangle=\sum_{\bm{a} \in \mathbb{F}_{q}^{n}}v_{\bm{a}}|\bm{a}\rangle,$$
where $\{|\bm{a}\rangle=|a_{1}\rangle\otimes \cdots \otimes |a_{n}\rangle: (a_{1},\ldots,a_{n})\in \mathbb{F}_{q}^{n}\}$ is a basis of $\mathbb{C}^{q^{n}}$ and $v_{\bm{a}}\in \mathbb{C}$. For any two $n$-qubits $|\bm{u}\rangle=\sum_{\bm{a} \in \mathbb{F}_{q}^{n}}u_{\bm{a}}|\bm{a}\rangle$ and $|\bm{v}\rangle=\sum_{\bm{a} \in \mathbb{F}_{q}^{n}}v_{\bm{a}}|\bm{a}\rangle$, their \emph{Hermitian inner product} is defined as
$$\langle \bm{u}|\bm{v}\rangle=\sum_{\bm{a}\in\mathbb{F}_{q}^{n}}u_{\bm{a}}\bar{v}_{\bm{a}} \in \mathbb{C},$$
where $\bar{v}_{\bm{a}}$ is the conjugate of $v_{\bm{a}}$ in the complex field. $|\bm{u}\rangle$ and $|\bm{v}\rangle$ are called distinguishable if $\langle \bm{u}|\bm{v}\rangle=0$.

The quantum errors in a quantum system are some
unitary operators, usually denoted $X$ and $Z$. The actions of $X(\bm{a})$ and $Z(\bm{b})$ on the basis $|\bm{v}\rangle\in \mathbb{C}^{q^{n}}$ are defined as
$$X(\bm{a})|\bm{v}\rangle=|\bm{v}+\bm{a}\rangle \;\textnormal{and} \; Z(\bm{b})|\bm{v}\rangle=\omega_{p}^{tr(\langle\bm{b},\bm{v}\rangle)}|\bm{v}\rangle$$
respectively, where $tr(\cdot)$ is the trace function from $\mathbb{F}_{q}$ to $\mathbb{F}_{p}$ ($p$ is the characteristic  of $\mathbb{F}_{q}$) and $\omega_{p}$ is a complex primitive $p$-th root of unity. The set of error operators on $\mathbb{C}^{q^{n}}$ is defined as
$$ E_{n}=\{ \omega_{p}^{i}X(\bm{a})Z(\bm{b}): 0\leq i\leq p-1,  \; \bm{a}=(a_{1},\ldots,a_{n}), \bm{b}=(b_{1},\ldots,b_{n}) \in \mathbb{F}_{q}^{n}\}.$$
For any error $E=\omega_{p}^{i}X(\bm{a})Z(\bm{b})$, we define
the \emph{quantum weight} of $E$ by
$$w_{Q}(E)=\sharp \{i:(a_{i},b_{i})\neq (0,0)\}.$$

Let $E_{n}(l)=\{E \in E_{n}:w_{Q}(E)\leq l\}$ be the set of error operators with weight no more than $l$. A quantum code with length $n$ is defined as a subspace of $\mathbb{C}^{q^{n}}$. A quantum code $Q$ can detect a quantum error $E$ if and only if for any $|\bm{u}\rangle,|\bm{v}\rangle \in Q$ with $\langle\bm{u}|\bm{v}\rangle=0$,  we have $\langle\bm{u}|E|\bm{v}\rangle=0$. The quantum code $Q$ has \emph{minimum distance} $d$ if $d$ is the largest integer such that for any  $|\bm{u}\rangle,|\bm{v}\rangle \in Q$ with $\langle\bm{u}|\bm{v}\rangle=0$ and $E\in E_{n}(d-1)$, we have $\langle\bm{u}|E|\bm{v}\rangle=0$. We denote by $((n,K,d))_{q}$ or $[[n, k, d]]_{q}$ a quantum code $Q$ of length $n$, dimension $K$ and minimum distance $d$, where $k=\log _{q}K$. A quantum code $Q$ is called a \emph{pure quantum} code if and only if for any $|\bm{u}\rangle,|\bm{v}\rangle \in Q$ and $E\in E_{n}$ with $1\leq w_{Q}(E)\leq d-1$ ($d$ is the minimum distance), we have $\langle\bm{u}|E|\bm{v}\rangle=0$.

Quantum codes can be constructed by using character theory of finite abelian groups from CSS construction. Suppose $S$ is an abelian subgroup of $E_{n}$,
\emph{quantum stabilizer} code $C(S)$ associated with $S$ is defined as
$$C(S)=\{|\psi\rangle\in\mathbb{C}^{q^{n}}: E|\psi\rangle=|\psi\rangle,\forall E \in S\}.$$
Quantum stabilizer codes are analogues of classical additive codes, and classical linear codes with certain orthogonality can be used to construct quantum stabilizer codes. More results and details can be found in \cite{1,2,3,4,16}.

In the asymmetric quantum channels, we require quantum codes to have different error-correcting capabilities for handling different types of errors. More research of quantum codes in the asymmetric quantum channels can be found in \cite{17,18,19}.

For any error $E=\omega_{p}^{i}X(\bm{a})Z(\bm{b})$, we define $X$-weight $w_{X}(E)$ and $Z$-weight $w_{Z}(E)$ as $w_{X}(E)=\sharp \{i:a_{i}\neq 0\}$ and $w_{Z}(E)=\sharp \{i:b_{i}\neq 0\}$, respectively. A quantum code $Q$ is called an \emph{asymmetric entanglement-assisted quantum error-correcting} code (AEAQECC) with parameters $[[n, k, d_{z}/d_{x}, c]]_{q}$ if $Q$ encodes $k$ logical qubits into $n$ physical qubits with the help of $c$ copies of maximally entangled states, which can detect all phase-flip errors ($Z$-errors) up to $d_{z}-1$ and all qubit-flip errors ($X$-errors) up to $d_{x}-1$. Namely, if $\langle\bm{u}|\bm{v}\rangle=0$ for $|\bm{u}\rangle$,$|\bm{v}\rangle \in Q$, then $\langle\bm{u}|E|\bm{v}\rangle=0$ for any $E \in E_{n}$ such that $w_{X}(E)\leq d_{x}-1$ and $w_{Z}(E)\leq d_{z}-1$. In \cite{12}, Galindo et al. gave the following construction from CSS construction.
\begin{theorem}\cite{12}
Let $C_{i}$ be an $[n,k_{i}]_{q}$ linear code with generator matrices $G_{i}$ for $i=1,2$. Set $d_{z}=wt\big(C_{1}^{\bot}\setminus (C_{2}\cap C_{1}^{\bot})\big)$ and
$d_{x}=wt\big(C_{2}^{\bot}\setminus (C_{1}\cap C_{2}^{\bot})\big)$. Then there exists an AEAQECC with parameters $[[n,n-k_{1}-k_{2}+c,d_{z}/d_{x},c]]_{q}$, where $c=\textnormal{rank}(G_{1}G_{2}^{T})=\dim(C_{1})-\dim(C_{1}\cap C_{2}^{\bot})$ is the minimum required of  maximally entangled states.
\end{theorem}

Let $Q$ be an AEAQECC with parameters $[[n,k,d_{z}/d_{x},c]]_{q}$ (where $k=n-k_{1}-k_{2}+c$) constructed by linear codes $C_{1}$ and $C_{2}$ with parameters $[n,k_{1}]_{q}$ and $[n,k_{2}]_{q}$ respectively. Then $Q$ is called a \emph{pure} AEAQECC if $d_{z}=wt\big(C_{1}^{\bot}\setminus (C_{2}\cap C_{1}^{\bot})\big)=wt(C_{1}^{\bot})$ and
$d_{x}=wt\big(C_{2}^{\bot}\setminus (C_{1}\cap C_{2}^{\bot})\big)=wt(C_{2}^{\bot})$. For a pure AEAQECC $Q$, by the Singleton bound of classical linear codes, we have $ d_{z}=wt(C_{1}^{\bot})\leq k_{1}+1$  and $d_{x}= wt(C_{2}^{\bot})\leq k_{2}+1$. It follows that
$$d_{x}+d_{z}\leq wt(C_{1}^{\bot})+wt(C_{2}^{\bot})\leq n-(n-k_{1}-k_{2}+c)+c+2=n-k+c+2.$$
Then $Q$ is called a pure \emph{maximum distance separable }(MDS) AEAQECC if the parameters satisfy $d_{x}+d_{z}=n-k+c+2$. In this paper, our purpose is to construct pure MDS AEAQECCs for all possible parameters.
\section{All possible linear $l$-intersection pairs of MDS codes}
In \cite{14}, Guenda et al. gave some results about linear $l$-intersection pairs of MDS codes over $\mathbb{F}_{q}$ with length $n\leq q+1$. In this section, we complement their results and give all possible linear $l$-intersection pairs of MDS codes over $\mathbb{F}_{2^{m}}$ with length $n=2^{m}+2\geq6$. In summary, linear $l$-intersection pairs of MDS codes for all possible parameters will be given in this section.

Let's begin with some basic results about MDS codes. Trivial families of MDS codes include the vector space $\mathbb{F}_{q}^{n}$, the codes equivalent to the
$[n,1,n]_{q}$ repetition codes and their duals $[n,n-1,2]_{q}$ codes for $n \geq 2$. The MDS Conjecture is given as follows.
\begin{conjecture}[MDS Conjecture]
If there is a nontrivial $[n,k,d]_{q}$ MDS code, then $n \leq q+1$, except when $q$ is even and $k=3$ or $k=q-1$ in which case $n\leq q+2$.
\end{conjecture}
Note that when $C_{1}$ is $\mathbb{F}_{q}^{n}$, it's easy to find that $\dim({C_{1} \cap C_{2}})= \dim(C_{2})$. Therefore, we assume that the dimensions of $C_{1}$ and $C_{2}$ are both less than $n$ in this paper.

First, let's recall some results in \cite{14}.  Guenda et al. used the definition of the extended GRS codes as follows. For polynomials $a(x)=a_{0}+a_{1}x+\cdots+a_{t}x^{t}$ and $b(x)=b_{0}+b_{1}x+\cdots+b_{t}x^{t}$ in  $\mathbb{F}_{q}[x]$ with $b_{t}\neq 0$, let $r(x)=\frac{a(x)}{b(x)}$ be a rational function. The evaluation $r(\infty)$ is defined to be $\frac{a_{t}}{b_{t}}$, thus $r(\infty)=0$ if and only if $\deg(a(x))<\deg(b(x))$. Let  $\mathcal{A}=\{a_{1},a_{2},\ldots,a_{n-1} \}$ be a subset of $\mathbb{F}_{q}$ with $n-1$ distinct elements, $\bm{v} =(v_{1},v_{2},\ldots,v_{n})$ where $v_{1},v_{2},\ldots,v_{n}$ are nonzero elements in $\mathbb{F}_{q}$ and $P(x)$ be a nonzero polynomial in $\mathbb{F}_{q} [x]$ with $\deg(P(x))\leq n$ such that $P(a_{i})\neq0$ for all $i=1,\ldots, n-1$. Then the extended GRS code is defined as
$$ GRS_{\infty}(\mathcal{A},P(x),\bm{v})=\{(\frac{v_{1}f(a_{1})}{P(a_{1})},\ldots ,\frac{v_{n-1}f(a_{n-1})}{P(a_{n-1})}, v_{n}(\frac{xf}{P})(\infty)) :$$
$$\;\;\;\;\;\;\;\;\;\;\;\;\;\;\;\;\;\;\;\;\;\;\;\;\;\;\;\;\;\;\; f(x) \in \mathbb{F}_{q}[x], \deg (f(x)) < \deg (P(x)) \}  .$$
It is well-known that $GRS_{\infty}(\mathcal{A},P(x),\bm{v})$ is an MDS code and so is its dual code. Then Guenda et
al. \cite{14} gave the following lemma for the existence of linear $l$-intersection pairs
of MDS codes.
\begin{lemma}\cite[Theorem 3.2 and Corollary 3.3]{14}
Let $q$ be a prime power and $n, k_{1}, k_{2}, l$ be non-negative integers such
that $k_{1} \leq n \leq q +1$ and $k_{2} \leq n$. For a subset $\mathcal{A} \subseteq \mathbb{F}_{q}$ of size $n-1$, if there exist polynomials $P(x)$, $Q(x)$
and $L(x)$ in $\mathbb{F}_{q} [x]$ satisfying the following
conditions:\\
(i) $\deg(P(x))=k_{1}$, $\deg(Q(x))=k_{2}$ and $\deg(L(x))=l$,\\
(ii) $\gcd(P(x), Q(x))=L(x)$,\\
(iii) $\gcd(P(x)Q(x),\prod_{a \in \mathcal{A}}(x-a))=1$,\\
(iv) $\deg(P(x))+\deg(Q(x)) \leq n+\deg(L(x))$,\\
then $GRS_{\infty}(\mathcal{A},P(x),\bm{v})$ and $GRS_{\infty}(\mathcal{A},Q(x),\bm{v})$ $($where $\bm{v} \in (\mathbb{F}_{q}^{\ast})^{n}$$)$ form a linear $l$-intersection pair of two MDS codes with parameters $[n, k_{1}, n-k_{1}+1]_{q}$ and $[n, k_{2}, n-k_{2}+1]_{q}$.
\end{lemma}
In \cite[pp. 92-93]{20}, the number of monic irreducible polynomials of degree $n$ over $\mathbb{F}_{q}$ is given as follows.
$$N_{q}(n)=\frac{1}{n}\sum_{d|n}\mu(d)q^{n/d},$$
where $\mu$ is the\emph{ M$\ddot{o}$bius function} defined by
$$\mu(m)=\begin{cases}
1\;\;\;\;\; \;\;\;\;\;\;\textnormal{if  m=1,}\\
(-1)^{r}\;\;\;\; \textnormal{if $m$ is a product of $r$ distinct primes,}\\
0\;\;\;\;\;\;\;\;\;\; \textnormal{ if $p^{2}|m$ for some prime $p$.}
\end{cases}$$
For all prime powers $q\geq3$, it's easy to find that $N_{q}(n)\geq 3$. Then Guenda et al. \cite{14} gave the following proposition to construct linear $l$-intersection pairs of MDS codes over $\mathbb{F}_{q}$ with length up to  $q+1$.
\begin{proposition}\cite[Proposition 3.1]{14}
Let $q\geq3$ be a prime power and $n, k_{1}, k_{2}, l$ be non-negative integers such
that $k_{1} \leq n-1 \leq q$ and $k_{2} \leq n-1$. If $ l \leq \min\{k_{1}, k_{2}\}$, then there exists a linear $l$-intersection
pair of MDS codes with parameters $[n, k_{1},n-k_{1}+1]_{q}$ and $[n, k_{2}, n-k_{2}+1]_{q}$.
\end{proposition}
\begin{proof}\cite[Proof of Proposition 3.1]{14}
Assume that $ l \leq \min\{k_{1}, k_{2}\}$. By  $N_{q}(n)\geq 3$, there exist monic irreducible polynomials
$f (x), L(x)$ and $h(x)$ in $\mathbb{F}_{q}$ of degrees $k_{1}-l$, $l$ and $k_{2}-l$, respectively, and the polynomial is set to be $1$ if the degree is zero. For any subset $\mathcal{A}\subseteq\mathbb{F}_{q}$ of size $n-1$, let $P(x)=f (x)L(x)$ and $Q(x)=h(x)L(x)$ so then
$P(x), Q(x)$ and $ L(x)$ satisfy the conditions in Lemma 3. Hence, $GRS_{\infty}(\mathcal{A},P(x),\bm{v})$ and $GRS_{\infty}(\mathcal{A},Q(x),\bm{v})$ form a linear $l$-intersection pair and they have parameters $[n, k_{1}, n-k_{1}+1]_{q}$ and $[n, k_{2}, n-k_{2}+1]_{q}$ respectively. \qed
\end{proof}
\begin{remark}
However, the above proof of \cite{14} is incomplete. We find that the above proof does not work in the following two cases:

\textbf{Case 1}: $n=q$, $k_{1},k_{2}\leq q-1$, $k_{1}-l=k_{2}-l=1$. From the above proof, if $n=q$, then $\mathcal{A}=\mathbb{F}_{q}\setminus \{\alpha\}$ for some $\alpha \in \mathbb{F}_{q}$. Hence, there exists only one monic irreducible polynomial $f(x)=x-\alpha\in\mathbb{F}_{q}[x]$ of degree $1$ satisfying  $\gcd(f(x),\prod_{a \in \mathcal{A}}(x-a))=1$. But in the case of $k_{1}-l=k_{2}-l=1$, we need two monic irreducible polynomials $f(x),h(x)$ of degree $1$ satisfying $\gcd(f(x),h(x))=1$ and $\gcd(f(x)h(x),\prod_{a \in \mathcal{A}}(x-a))=1$, which leads to a contradiction.

\textbf{Case 2}: $n=q+1$, $k_{1},k_{2}\leq q$, $1 \in \{ l, k_{1}-l, k_{2}-l \}$. From the above proof, if $n=q+1$, then $\mathcal{A}=\mathbb{F}_{q}$. Hence, there does not exist monic irreducible polynomial $f(x)\in\mathbb{F}_{q}[x]$ of degree $1$ satisfying $\gcd(f(x),\prod_{a \in \mathcal{A}}(x-a))=1$. Therefore, if $n=q+1$ and $1 \in \{ l, k_{1}-l, k_{2}-l \}$, the above proof does not work.
\end{remark}
\subsection{Complement of linear $l$-intersection pairs of MDS codes over $\mathbb{F}_{q}$ with length $n\leq q+1$}
In the following, we give the constructions of the two cases in Remark 1.
\begin{theorem}
Let $q\geq3$ be a prime power and $k_{1}, k_{2}, l$ be non-negative integers such that $k_{1}=k_{2}=l+1\leq q-1$. Then there exists a linear $l$-intersection pair of two MDS codes with the same parameters $[q,l+1,q-l]_{q}$.
\end{theorem}
\begin{proof}
Write $\mathbb{F}_{q}=\{a_{1},\ldots,a_{q-1},0\}$ and $\mathbb{F}_{q}^{\ast}=\{a_{1},\ldots,a_{q-1}\}$, then we divide our proof into two cases.

\textbf{Case 1}: When $0<l\leq q-2$. Let $C_{1}=GRS_{l+1}(\mathbb{F}_{q},\bm{1})$ with generator matrix $G_{1}$ and $C_{2}=GRS_{l+1}(\mathbb{F}_{q}^{\ast}\cup \infty,\bm{1})$ with generator matrix $G_{2}$, where
$$G_{{1}}=\left(
  \begin{array}{cccc}
    1 &  \ldots & 1 &  1 \\
       a_{1} & \ldots  & a_{q-1} & 0 \\
     \vdots & \ddots & \vdots & \vdots \\
     a_{1}^{l} & \ldots  & a_{q-1}^{l} & 0 \\
  \end{array}
\right)
,
G_{{2}}=\left(
  \begin{array}{cccc}
    1 &  \ldots & 1 &  0 \\
       a_{1} & \ldots  & a_{q-1} & 0 \\
     \vdots & \ddots & \vdots & \vdots \\
     a_{1}^{l} & \ldots  & a_{q-1}^{l} & 1 \\
  \end{array}
\right)
.$$
Let $S=\textnormal{span}_{\mathbb{F}_{q}}\{(a_{1},\ldots,a_{q-1},0),\ldots ,(a_{1}^{l-1},\ldots,a_{q-1}^{l-1},0) , (a^{l}_{1}+1,\ldots,a^{l}_{q-1}+1,1)\}$. It is easy to see that $\dim(S)=l$ and $S\subseteq C_{1} \cap C_{2}$, thus $\dim (C_{1} \cap C_{2})\geq \dim(S)=l$. Note that $(1,\ldots,1,1) \in C_{1}$ but not in $C_{2}$, thus $\dim(C_{1} \cap C_{2}) \leq \dim(C_{1})-1=l$. Hence, $\dim (C_{1} \cap C_{2})=l$, then there exists a linear $l$-intersection
pair of two MDS codes with the same parameters $[q, l+1,q-l]_{q}$.

\textbf{Case 2}: When $l=0$.  Let $C_{1}=\mathbb{F}_{q}\cdot\bm{c}_{1}$ and $C_{2}=\mathbb{F}_{q}\cdot\bm{c}_{2}$ where $\bm{c}_{1}=(1,\ldots,1)$ and $\bm{c}_{2}=(a_{1},\ldots,a_{q-1},1)$ respectively. Then $C_{1}$ and $C_{2}$ are MDS codes with the same parameters $[q,1,q]_{q}$ satisfying $C_{1} \cap C_{2}=\{\bm{0}\}$, i.e., $\dim(C_{1} \cap C_{2})=0$.\qed
\end{proof}
\begin{theorem}
Let $q\geq3$ be a prime power and $k_{1}, k_{2}, l$ be non-negative integers such that $k_{1},k_{2}\leq q$ and $\max \{k_{1}+k_{2}-q-1, 0 \} \leq l \leq \min \{k_{1}, k_{2}\}$. If $1 \in \{ l, k_{1}-l, k_{2}-l \}$ $($except $(k_{1},k_{2},l)\in\{(2,1,1),(1,2,1)\})$, then there exists a linear $l$-intersection pair of two MDS codes with parameters $[q+1, k_{1},q-k_{1}+2]_{q}$ and $[q+1, k_{2},q-k_{2}+2]_{q}$.
\end{theorem}
\begin{proof}
Let $\mathbb{F}_{q}=\{a_{1},\ldots,a_{q-1},0\}$, $\mathcal{A}_{1}=\{a_{1},\ldots,a_{q-1},0,\infty\}$, $\mathcal{A}_{2}=\{a_{1},\ldots,a_{q-1}, \\ \infty,0\}$. Without loss of generality, we assume that $k_{1} \leq k_{2}$.

\textbf{Case 1}:  $l=1$.

(i) When $k_{1}=k_{2}=1$. Let $C_{1}=C_{2}$ be two $[q+1,1,q+1]_{q}$ MDS codes, then $\dim(C_{1} \cap C_{2})=1$.

(ii) When $k_{1}=1$, $k_{2}=2$. We prove that there is no linear $l$-intersection pair of two MDS codes with parameters $[q+1,1,q+1]_{q}$ and $[q+1,2,q]_{q}$. Otherwise, suppose $C_{1}$ and $C_{2}$ are MDS codes satisfying  $\dim (C_{1} \cap C_{2})=1$ with parameters $[q+1,1,q+1]_{q}$ and $[q+1,2,q]_{q}$ respectively. Then $C_{1}\subseteq C_{2}$. However, for $C_{2}$, according to Theorem 1, $A_{q+1}=(q-1)\sum_{j=0}^{1}(-1)^{j} {q\choose j} q^{1-j}=0$,  i.e., $C_{2}$  has no codewords of weight $q+1$, which leads to a contradiction.

(iii) When $k_{1}=1$, $3 \leq k_{2}\leq q$. Let $C_{2}$ be a $[q+1,k_{2},q-k_{2}+2]_{q}$ MDS code. For $C_{2}$, according to Theorem 1, $A_{q+1}=(q-1)\sum_{j=0}^{k_{2}-1}(-1)^{j} {q\choose j} q^{k_{2}-1-j}$. For $j >0$, note that ${q\choose j}/ {q\choose j+1}=\frac{j+1}{q-j}>\frac{1}{q}$, thus ${q\choose j}q^{k_{2}-1-j}-{q\choose j+1}q^{k_{2}-1-(j+1)}>0$. Therefore, for $k_{2}\geq3$, when $k_{2}$ is even,

$\;\;\;\;A_{q+1}=(q-1)\sum_{m=0}^{\frac{k_{2}-2}{2}}\big({q\choose 2m}q^{k_{2}-1-2m}-{q\choose 2m+1}q^{k_{2}-1-(2m+1)}\big)>0;$
\\When $k_{2}$ is odd, \\
$A_{q+1}=(q-1)\Big(\sum_{m=0}^{\frac{k_{2}-3}{2}}\big({q\choose 2m}q^{k_{2}-1-2m}-{q\choose 2m+1}q^{k_{2}-1-(2m+1)}\big)+{q\choose k_{2}-1} q^{0}\Big)>0.$
Thus  $A_{q+1}>0$ for $k_{2}\geq3$, i.e., there exists a $\bm{c}\in C_{2}$ with $wt(\bm{c})=q+1$. Let $C_{1}=\mathbb{F}_{q}\cdot \bm{c}$ with parameters $[q+1,1,q+1]_{q}$, then $\dim(C_{1} \cap C_{2})=1$.

(iv) When $2\leq k_{1} \leq k_{2}\leq q$, $k_{1}+k_{2}\leq q$. Let $\bm{v}=( a_{1}^{k_{1}-1}, \ldots  ,a_{q-1}^{k_{1}-1}, 1, 1)$, $C_{1}=GRS_{k_{1}}(\mathcal{A}_{1}, \bm{1})$ with generator matrix $G_{1}$ and $C_{2}=GRS_{k_{2}}(\mathcal{A}_{2}, \bm{v})$ with generator matrix $G_{2}$, where
$$G_{{1}}=\left(
  \begin{array}{ccccc}
    1 &  \ldots & 1 &  1 &  0 \\
       a_{1} & \ldots  & a_{q-1} & 0 & 0\\
     \vdots & \ddots & \vdots & \vdots & \vdots \\
   a_{1}^{k_{1}-2} & \ldots  & a_{q-1}^{k_{1}-2} & 0 & 0\\
   a_{1}^{k_{1}-1} & \ldots  & a_{q-1}^{k_{1}-1} & 0 & 1\\
  \end{array}
\right)
, G_{{2}}=\left(
  \begin{array}{ccccc}
   a_{1}^{k_{1}-1} & \ldots  & a_{q-1}^{k_{1}-1} & 0 & 1 \\
       a_{1}^{k_{1}} & \ldots  & a_{q-1}^{k_{1}} & 0 & 0\\
     \vdots &\ddots & \vdots & \vdots & \vdots \\
     a_{1}^{k_{1}+k_{2}-3} & \ldots  & a_{q-1}^{k_{1}+k_{2}-3} & 0 & 0\\
 a_{1}^{k_{1}+k_{2}-2} & \ldots  & a_{q-1}^{k_{1}+k_{2}-2} & 1 & 0\\
  \end{array}
\right)
.$$
Let $S=\textnormal{span}_{\mathbb{F}_{q}}\{(a_{1}^{k_{1}-1},\ldots,a_{q-1}^{k_{1}-1},0,1)\}$, then $\dim(S)=1$ and $S\subseteq C_{1}\cap C_{2}$. Thus $\dim(C_{1}\cap C_{2})\geq\dim(S)=1$.
Let $T=\textnormal{span}_{\mathbb{F}_{q}}\{(1,\ldots,1,0),(a_{1},\ldots,a_{q-1},0,0),\ldots,\\(a_{1}^{k_{1}-2},\ldots,a_{q-1}^{k_{1}-2},0,0)\}$, then $\dim(T)=k_{1}-1$ and $T\subseteq C_{1}$. Since $k_{1}+k_{2}-2\leq q-2$, the basis of $T$ is linearly independent with the rows of $G_{2}$, i.e., $T \cap C_{2}=\{\bm{0}\}$. Hence, $\dim(C_{1}\cap C_{2})\leq \dim(C_{1})-\dim(T)=1$. Then  $\dim(C_{1}\cap C_{2})=1$.

(v) When $2\leq k_{1} \leq k_{2}\leq q$, $k_{1}+k_{2}= q+1$. Let $\bm{v}=( a_{1}^{k_{1}-1}, \ldots  ,a_{q-1}^{k_{1}-1}, 1, 1)$, $C_{1}=GRS_{k_{1}}(\mathcal{A}_{1}, \bm{1})$ with generator matrix $G_{1}$ and $C_{2}=GRS_{k_{2}}(\mathcal{A}_{1}, \bm{v})$ with generator matrix $G_{2}$, where
$$G_{{1}}=\left(
  \begin{array}{ccccc}
    1 &  \ldots & 1 &  1 &  0 \\
       a_{1} & \ldots  & a_{q-1} & 0 & 0\\
     \vdots & \ddots & \vdots & \vdots & \vdots \\
     a_{1}^{k_{1}-1} & \ldots  & a_{q-1}^{k_{1}-1} & 0 & 1\\
  \end{array}
\right)
, G_{{2}}=\left(
  \begin{array}{ccccc}
   a_{1}^{k_{1}-1} & \ldots  & a_{q-1}^{k_{1}-1} & 1 & 0 \\
       a_{1}^{k_{1}} & \ldots  & a_{q-1}^{k_{1}} & 0 & 0\\
     \vdots & \ddots& \vdots & \vdots & \vdots \\
     a_{1}^{q-1} & \ldots  & a_{q-1}^{q-1} & 0 & 1\\
  \end{array}
\right)
.$$
Let $S=\textnormal{span}_{\mathbb{F}_{q}}\{(a_{1}^{k_{1}-1}+1 , \ldots  ,a_{q-1}^{k_{1}-1}+1 , 1 , 1)\}$, then $\dim(S)=1$ and $S\subseteq C_{1}\cap C_{2}$. Thus $\dim(C_{1}\cap C_{2})\geq\dim(S)=1$.
Let $T=\textnormal{span}_{\mathbb{F}_{q}}\{(1,\ldots,1,0),(a_{1},\ldots,a_{q-1},0,0), \\ \ldots,(a_{1}^{k_{1}-2},\ldots,a_{q-1}^{k_{1}-2},0,0)\}$, then $\dim(T)=k_{1}-1$ and $T\subseteq C_{1}$. It's easy to see that the basis of $T$ is linearly independent with the rows of $G_{2}$, i.e., $T \cap C_{2}=\{\bm{0}\}$. Hence, $\dim(C_{1}\cap C_{2})\leq \dim(C_{1})-\dim(T)=1$. Then  $\dim(C_{1}\cap C_{2})=1$.

(vi) When $2\leq k_{1} \leq k_{2}\leq q$, $k_{1}+k_{2}= q+2$ (By Lemma 1, $k_{1}+k_{2}$ is no more than $q+2$). Let $\bm{v}=( a_{1}^{k_{1}-2}, \ldots  ,a_{q-1}^{k_{1}-2}, 1, 1)$,  $C_{1}=GRS_{k_{1}}(\mathcal{A}_{1}, \bm{1})$ with generator matrix $G_{1}$ and $C_{2}=GRS_{k_{2}}(\mathcal{A}_{1}, \bm{v})$ with generator matrix $G_{2}$, where
$$G_{{1}}=\left(
  \begin{array}{ccccc}
    1 &  \ldots & 1 &  1 &  0 \\
       a_{1} & \ldots  & a_{q-1} & 0 & 0\\
     \vdots & \ddots & \vdots & \vdots & \vdots \\
     a_{1}^{k_{1}-2} & \ldots  & a_{q-1}^{k_{1}-2} & 0 & 0\\
     a_{1}^{k_{1}-1} & \ldots  & a_{q-1}^{k_{1}-1} & 0 & 1\\
  \end{array}
\right)
, G_{{2}}=\left(
  \begin{array}{ccccc}
   a_{1}^{k_{1}-2} & \ldots  & a_{q-1}^{k_{1}-2} & 1 & 0 \\
       a_{1}^{k_{1}-1} & \ldots  & a_{q-1}^{k_{1}-1} & 0 & 0\\
     \vdots & \ddots & \vdots & \vdots & \vdots \\
      a_{1}^{q-2} & \ldots & a_{q-1}^{q-2} & 0 & 0 \\
     a_{1}^{q-1} & \ldots  & a_{q-1}^{q-1} & 0 & 1\\
  \end{array}
\right)
.$$
Let $S=\textnormal{span}_{\mathbb{F}_{q}}\{(a_{1}^{k_{1}-1}+a_{1}^{k_{1}-2}+1 , \ldots  , a_{q-1}^{k_{1}-1}+a_{q-1}^{k_{1}-2}+1 ,1, 1 )\}$, then $\dim(S)=1$ and $S\subseteq C_{1}\cap C_{2}$. Thus $\dim(C_{1}\cap C_{2})\geq\dim(S)=1$.
Let $T=\textnormal{span}_{\mathbb{F}_{q}}\{(1,\ldots,1,0),\\(a_{1},\ldots,a_{q-1},0,0),\ldots,(a_{1}^{k_{1}-2},\ldots,a_{q-1}^{k_{1}-2},0,0)\}$, then $\dim(T)=k_{1}-1$ and $T\subseteq C_{1}$. It's easy to see that the basis of $T$ is linearly independent with the rows of $G_{2}$, i.e., $T \cap C_{2}=\{\bm{0}\}$. Hence, $\dim(C_{1}\cap C_{2})\leq \dim(C_{1})-\dim(T)=1$. Then  $\dim(C_{1}\cap C_{2})=1$.

\textbf{Case 2}: $k_{1}-l=1$, $l\geq 2$.

(i) When $k_{1}<k_{2}\leq q$. Let $C_{1}=GRS_{k_{1}}(\mathcal{A}_{1}, \bm{1})$ with generator matrix $G_{1}$ and $C_{2}=GRS_{k_{2}}(\mathcal{A}_{1}, \bm{1})$ with generator matrix $G_{2}$, where
$$G_{{1}}=\left(
  \begin{array}{ccccc}
    1 &  \ldots & 1 &  1 &  0 \\
       a_{1} & \ldots  & a_{q-1} & 0 & 0\\
     \vdots & \ddots & \vdots & \vdots & \vdots \\
     a_{1}^{k_{1}-1} & \ldots  & a_{q-1}^{k_{1}-1} & 0 & 1\\
  \end{array}
\right)
, G_{{2}}=\left(
  \begin{array}{ccccc}
    1 &  \ldots & 1 &  1 &  0 \\
       a_{1} & \ldots  & a_{q-1} & 0 & 0\\
     \vdots & \ddots & \vdots & \vdots & \vdots \\
     a_{1}^{k_{2}-1} & \ldots  & a_{q-1}^{k_{2}-1} & 0 & 1\\
  \end{array}
\right)
.$$
Note that the first $k_{1}-1$ rows of $G_{1}$ are also rows of $G_{2}$, thus $\dim(C_{1}\cap C_{2})\geq k_{1}-1$. If the $k_{1}$-th row of $G_{1}$ is belonged to $C_{2}$, then  $(a_{1}^{k_{1}-1},\ldots, a_{q-1}^{k_{1}-1},0 ,1)-(a_{1}^{k_{1}-1},\ldots, a_{q-1}^{k_{1}-1},0 ,0)=(0,\ldots,0,0 ,1) \in C_{2}$, however, $d(C_{2})=q+1-k_{2}+1\geq q+1-q+1=2$, which leads to a contradiction. Hence, $\dim(C_{1}\cap C_{2})\leq \dim(C_{1})-1= k_{1}-1$. Then  $\dim(C_{1}\cap C_{2})= k_{1}-1=l$.

(ii) When $k_{1}=k_{2} \leq q-1$.  Let $C_{1}=GRS_{l+1}(\mathcal{A}_{1}, \bm{1})$ with generator matrix $G_{1}$ and $C_{2}=GRS_{l+1}(\mathcal{A}_{2}, \bm{1})$ with generator matrix $G_{2}$, where
$$G_{{1}}=\left(
  \begin{array}{ccccc}
    1 &  \ldots & 1 &  1 &  0 \\
       a_{1} & \ldots  & a_{q-1} & 0 & 0\\
     \vdots & \ddots & \vdots & \vdots & \vdots \\
     a_{1}^{l} & \ldots  & a_{q-1}^{l} & 0 & 1\\
  \end{array}
\right)
, G_{{2}}=\left(
  \begin{array}{ccccc}
    1 &  \ldots & 1 &  0 &  1 \\
       a_{1} & \ldots  & a_{q-1} & 0 & 0\\
     \vdots & \ddots & \vdots & \vdots & \vdots \\
     a_{1}^{l} & \ldots  & a_{q-1}^{l} & 1 & 0\\
  \end{array}
\right)
.$$
Let $S=\textnormal{span}_{\mathbb{F}_{q}}\{( a_{1} , \ldots  , a_{q-1} , 0 , 0),\ldots,  (a_{1}^{l-1} , \ldots  , a_{q-1}^{l-1} , 0 , 0),  (a_{1}^{l}+1 , \ldots  , a_{q-1}^{l}+1 , 1 , 1)\}$. Since $l\leq q-2$, it follows that $\dim(S)=l$ and $S\subseteq C_{1} \cap C_{2}$,  hence, $\dim(C_{1}\cap C_{2})\geq \dim(S)=l$. Obviously, $(1,1,\ldots,1,0) \notin C_{2}$, hence, $\dim(C_{1}\cap C_{2})\leq \dim(C_{1})-1=l$. Then  $\dim(C_{1}\cap C_{2})= l$.

(iii) When $k_{1}=k_{2}=q$. Let $\bm{v}=( a_{1}, \ldots  ,a_{q-1}, 1, 1)$, $C_{1}=GRS_{q}(\mathcal{A}_{1}, \bm{1})$ with generator matrix $G_{1}$ and $C_{2}=GRS_{q}(\mathcal{A}_{1}, \bm{v})$ with generator matrix $G_{2}$, where
$$G_{{1}}=\left(
  \begin{array}{ccccc}
    1 &  \ldots & 1 &  1 &  0 \\
       a_{1} & \ldots  & a_{q-1} & 0 & 0\\
     \vdots & \ddots & \vdots & \vdots & \vdots \\
     a_{1}^{q-1} & \ldots  & a_{q-1}^{q-1} & 0 & 1\\
  \end{array}
\right)
, G_{{2}}=\left(
  \begin{array}{ccccc}
    a_{1} & \ldots  & a_{q-1} & 1 &  0 \\
       a_{1}^{2} & \ldots  & a_{q-1}^{2} & 0 & 0\\
     \vdots & \ddots & \vdots & \vdots & \vdots \\
     a_{1}^{q-1} & \ldots  & a_{q-1}^{q-1} & 0 & 0\\
   a_{1} & \ldots  & a_{q-1} & 0 &  1 \\
\end{array}
\right)
.$$
Let $S=\textnormal{span}_{\mathbb{F}_{q}}\{( a_{1} , \ldots  , a_{q-1} , 0 , 0),\ldots,  (a_{1}^{q-2} , \ldots  , a_{q-1}^{q-2} , 0 , 0),  (0, \ldots  ,0 , 1 , -1)\}$. N-ote that $(0, \ldots  ,0 , 1 , -1)=( a_{1} , \ldots  , a_{q-1} , 1 , 0)-( a_{1} , \ldots  , a_{q-1} , 0 , 1)=(1, \ldots  ,1 , 1 , 0)\\-(a_{1}^{q-1} , \ldots  , a_{q-1}^{q-1} , 0 , 1)$, thus $S\subseteq C_{1} \cap C_{2}$ and $\dim(S)=q-1$.  Then $\dim(C_{1}\cap C_{2})\geq \dim(S)=q-1$. Obviously, $(1,1,\ldots,1,0) \notin C_{2}$, hence,   $\dim(C_{1}\cap C_{2})\leq \dim(C_{1})-1=q-1$. Then  $\dim(C_{1}\cap C_{2})= q-1=l$.

In summary, when $n=q+1$ with $1 \in \{ l, k_{1}-l, k_{2}-l \}$ for $k_{1},k_{2}\leq q$, we give constructions of all possible cases, then the theorem holds. \qed
\end{proof}
\subsection{Linear $l$-intersection pairs of MDS codes over $\mathbb{F}_{2^{m}}$ with length $n=2^{m}+2\geq6$}
In this section, assuming the validity of MDS Conjecture, we consider linear $l$-intersection pairs of two MDS codes with parameters $[n,k_{1},n-k_{1}+1]_q$ and $[n,k_{2},n-k_{2}+1]_q$, where  $q=2^{m}\geq4$, $n=q+2$ and $k_{1},k_{2} \in \{3, q-1\}$.

For $q=2^{m}\geq4$, let $\mathbb{F}_{q}=\{a_{1},\ldots,a_{q-1},0\}$ and $v_{1}, v_{2},\ldots v_{q+2}$ be nonzero elements in $\mathbb{F}_{q}$. Then there exist a $[q+2,3,q]_{q}$ MDS code with generator matrix $G_{1}$ or parity check matrix $G_{2}$ and a $[q+2,q-1,4]_{q}$ MDS code with generator matrix $G_{2}$ or parity check matrix $G_{1}$, where
$$G_{1}=\left(
  \begin{array}{cccccc}
    v_{1} &  \ldots & v_{q-1} &  v_{q} &  0 &  0 \\
 v_{1}a_{1} & \ldots  & v_{q-1}a_{q-1} & 0 & v_{q+1} & 0\\
    v_{1}a_{1}^{2} & \ldots  & v_{q-1}a_{q-1}^{2} & 0 & 0 & v_{q+2}\\
\end{array}
\right),
$$
$$G_{2}=\left(
    \begin{array}{ccccccc}
      v_{1} & v_{2}a_{1} & v_{3}a_{1}^{2} & v_{4} & 0 & \ldots & 0 \\
      v_{1} & v_{2}a_{2} & v_{3}a_{2}^{2} & 0 & v_{5} & \ldots & 0 \\
      \vdots & \vdots & \vdots  & 0 & 0 & \ddots & 0 \\
      v_{1} & v_{2}a_{q-1} & v_{3}a_{q-1}^{2} & 0 & 0 & \ldots & v_{q+2} \\
    \end{array}
  \right).
$$

Let $C_{1}$ and $C_{2}$ be two MDS codes with the same parameters  $[q+2,q-1,4]_{q}$. By Lemma 1, $q-4\leq\dim( C_{1} \cap C_{2})\leq q-1$, then we obtain the following theorem.
\begin{theorem}
For $q\geq4$ and $q-4 \leq l\leq q-1$, there exist  linear $l$-intersection pairs of two MDS codes with the same parameters $[q+2,q-1,4]_{q}$.
\end{theorem}
\begin{proof}
Let $\mathbb{F}_{q}^{\ast}=\{a_{1},\ldots,a_{q-1}\}$, $v_{i},v_{i}' \in \mathbb{F}_{q}^{\ast}$ for $i=4,\ldots,q+2$ and
$$U=\left(
    \begin{array}{ccc}
      1 & a_{1} & a_{1}^{2} \\
      1 & a_{2} & a_{2}^{2}  \\
      \vdots & \vdots & \vdots  \\
      1 & a_{q-1} & a_{q-1}^{2} \\
    \end{array}
  \right),
V=\left(
    \begin{array}{cccc}
      v_{4} & 0 & \ldots & 0 \\
     0 & v_{5} & \ldots & 0 \\
       0 & 0 & \ddots & 0 \\
      0 & 0 & \ldots & v_{q+2} \\
    \end{array}
  \right),
V'=\left(
    \begin{array}{cccc}
      v'_{4} & 0 & \ldots & 0 \\
     0 & v'_{5} & \ldots & 0 \\
       0 & 0 & \ddots & 0 \\
      0 & 0 & \ldots & v'_{q+2} \\
    \end{array}
  \right).
$$
Let $C_{1}$ and $C_{2}$ be two $[q+2,q-1,4]_{q}$ MDS codes with generator matrices $G_{1}$ and $G_{2}$ respectively, where
$$G_{1}=\left(
    \begin{array}{ccccccc}
      1 & a_{1} & a_{1}^{2} & v_{4} & 0 & \ldots & 0 \\
      1 & a_{2} & a_{2}^{2} &  0 & v_{5} & \ldots & 0 \\
      \vdots &\vdots & \vdots &  0 & 0 & \ddots & 0 \\
      1 & a_{q-1} & a_{q-1}^{2} &  0 & 0 & \ldots & v_{q+2} \\
    \end{array}
  \right)=(U\,|\,V) \;
\textnormal{and}\;G_{2}=(U\,|\,V').
$$
Note that $(U\,|\,V')\cdot {I\choose -V'^{-1}U }=U-V'V'^{-1}U=0$, thus the parity check matrix $H_{2}$ of $C_{2}$ is $(I\,|-U^{T}V'^{-1})$. Then $G_{1}H_{2}^{T}=(U\,|\,V)\cdot {I\choose -V'^{-1}U }=U-VV'^{-1}U$. Let $b_{i}=1-v_{i}\cdot v_{i}'^{-1}$ for $i=4,\ldots,q+2$, then
$$G_{1}H_{2}^{T}=\left(
    \begin{array}{ccc}
      b_{4} & b_{4}a_{1} & b_{4}a_{1}^{2} \\
      b_{5} & b_{5}a_{2} & b_{5}a_{2}^{2}\\
      \vdots & \vdots & \vdots\\
      b_{q+2} & b_{q+2}a_{q-1} & b_{q+2}a_{q-1}^{2}\\
    \end{array}
  \right).
$$
Obviously, for any $q-4 \leq l \leq q-1$, we can always choose $v_{i}=v_{i}'$ for $4\leq i\leq l+3$ and $b_{i}\neq0$ for $l+4\leq i\leq q+2$ such that $\textnormal{rank}(G_{1}H_{2}^{T})=q-1-l$. By Lemma 2, $\dim( C_{1} \cap C_{2})= \dim( C_{1})-\textnormal{rank}(G_{1}H_{2}^{T})=l$. Hence,  the theorem  holds. \qed
\end{proof}

Let $C_{1}$ and $C_{2}$ be two MDS codes with the same parameters  $[q+2,3,q]_{q}$. By Lemma 1, $0\leq\dim( C_{1} \cap C_{2})\leq 3$, then we obtain the following theorem.
\begin{theorem}
For $q>4$ and $0\leq l\leq3$, there exist linear $l$-intersection pairs of two MDS codes with the same parameters $[q+2,3,q]_{q}$.
\end{theorem}
\begin{proof}
(i) $l=0$: Let  $C_{1}$ and $C_{2}$ be two $[q+2,3,q]_{q}$ MDS codes with  generator matrices $G_{1}$ and $G_{2}$ respectively, where
$$G_{1}=\left(
  \begin{array}{cccccc}
    1 &  \ldots & 1 &  1 &  0 &  0 \\
 a_{1} & \ldots  & a_{q-1} & 0 & 1 & 0\\
    a_{1}^{2} & \ldots  & a_{q-1}^{2} & 0 & 0 & 1\\
\end{array}
\right)
,
G_{2}=\left(
  \begin{array}{cccccc}
    a_{1}^{3} & \ldots  & a_{q-1}^{3} & 1 & 0 & 0\\
 a_{1}^{4} & \ldots  & a_{q-1}^{4} & 0 & 1 & 0\\
    a_{1}^{5} & \ldots  & a_{q-1}^{5} & 0 & 0 & 1\\
\end{array}
\right).
$$
Obviously, $\dim(C_{1} \cap C_{2})=0$.

(ii) $l=1$: Let  $C_{1}$ and $C_{2}$ be two $[q+2,3,q]_{q}$ MDS codes with  generator matrices $G_{1}$ and $G_{2}$ respectively, where
$$G_{1}=\left(
  \begin{array}{cccccc}
    1 &  \ldots & 1 &  1 &  0 &  0 \\
 a_{1} & \ldots  & a_{q-1} & 0 & 1 & 0\\
    a_{1}^{2} & \ldots  & a_{q-1}^{2} & 0 & 0 & 1\\
\end{array}
\right)
,
G_{2}=\left(
  \begin{array}{cccccc}
    a_{1}^{2} & \ldots  & a_{q-1}^{2} & 0 & 0 & 1\\
 a_{1}^{3} & \ldots  & a_{q-1}^{3} & 0 & 1 & 0\\
    a_{1}^{4} & \ldots  & a_{q-1}^{4} & 1 & 0 & 0\\
\end{array}
\right).
$$
Note that $C_{1} \cap C_{2} = \textnormal{span}_{\mathbb{F}_{q}}\{ (a_{1}^{2} , \ldots  , a_{q-1}^{2} , 0 , 0 , 1)\}$, thus $\dim(C_{1} \cap C_{2})=1$.

(iii) $l=2$: Let  $C_{1}$ and $C_{2}$ be two $[q+2,3,q]_{q}$ MDS codes with  generator matrices $G_{1}$ and $G_{2}$ respectively, where
$$G_{1}=\left(
  \begin{array}{cccccc}
    1 &  \ldots & 1 &  1 &  0 &  0 \\
 a_{1} & \ldots  & a_{q-1} & 0 & 1 & 0\\
    a_{1}^{2} & \ldots  & a_{q-1}^{2} & 0 & 0 & 1\\
\end{array}
\right)
,
G_{2}=\left(
  \begin{array}{cccccc}
    1 &  \ldots & 1 &  0 &  0 &  1 \\
 a_{1} & \ldots  & a_{q-1} & 0 & 1 & 0\\
    a_{1}^{2} & \ldots  & a_{q-1}^{2} & 1 & 0 & 0\\
\end{array}
\right).
$$
Note that $S=\textnormal{span}_{\mathbb{F}_{q}}\{(a_{1}^{2}+1 , \ldots  , a_{q-1}^{2}+1 , 1 , 0 , 1),(a_{1} , \ldots  , a_{q-1} , 0 , 1 , 0)\}\subseteq C_{1} \cap C_{2}$, thus $\dim(C_{1} \cap C_{2})\geq \dim(S)=2$. Obviously, $(1 ,  \ldots , 1 , 1 ,  0 ,  0)\notin C_{2}$, hence, $\dim(C_{1} \cap C_{2})\leq \dim(C_{1})-1=2$. Then $\dim(C_{1} \cap C_{2})= 2$.

(iv) $l=3$: Let $C_{1}=C_{2}$ be two $[q+2,3,q]_{q}$ MDS codes. Obviously, we have  $\dim(C_{1} \cap C_{2})= 3$. \qed
\end{proof}

Let $C_{1}$ and $C_{2}$ be two MDS codes with parameters  $[q+2,3,q]_{q}$ and $[q+2,q-1,4]_{q}$ respectively. By Lemma 1, $0\leq\dim( C_{1} \cap C_{2})\leq 3$, then we obtain the following theorem.
\begin{theorem}
For $q>4$ and $0\leq l\leq3$, there exist  linear $l$-intersection pairs of two MDS codes with parameters $[q+2,3,q]_{q}$ and $[q+2,q-1,4]_{q}$.
\end{theorem}
\begin{proof}
(i) $l=0$:  Let $C_{1}$ be a $[q+2,3,q]_{q}$ MDS code with generator matrix $G_{1}$ and $C_{2}$ be a $[q+2,q-1,4]_{q}$ MDS code with parity check matrix  $H_{2}$, where
$$G_{1}=\left(
  \begin{array}{cccccc}
   a_{1} & \ldots  & a_{q-1} & 1 & 0 & 0\\
    a_{1}^{2} & \ldots  & a_{q-1}^{2} & 0 & 1 & 0\\
 a_{1}^{3} & \ldots  & a_{q-1}^{3} & 0 & 0 & 1\\
\end{array}
\right)
,
H_{2}=\left(
  \begin{array}{cccccc}
    1 & \ldots  & 1 & 1 & 0 & 0\\
 a_{1} & \ldots  & a_{q-1} & 0 & 1 & 0\\
    a_{1}^{2} & \ldots  & a_{q-1}^{2} & 0 & 0 & 1\\
\end{array}
\right)
,
G_{1}H_{2}^{T}=\left(
  \begin{array}{ccc}
    1 &  0 & 0  \\
    0 &  1  & 0 \\
    0 & 0  & 1 \\
\end{array}
\right).
$$
By Lemma 2, $\dim( C_{1} \cap C_{2})=\dim( C_{1})-\textnormal{rank}(G_{1}H_{2}^{T})=3-3=0$.

(ii) $l=1$:  Let $C_{1}$ be a $[q+2,3,q]_{q}$ MDS code with generator matrix $G_{1}$ and $C_{2}$ be a $[q+2,q-1,4]_{q}$ MDS code with parity check matrix  $H_{2}$, where
$$G_{1}=\left(
  \begin{array}{cccccc}
    1 &  \ldots & 1 &  1 &  0 &  0 \\
 a_{1} & \ldots  & a_{q-1} & 0 & 1 & 0\\
    a_{1}^{2} & \ldots  & a_{q-1}^{2} & 0 & 0 & 1\\
\end{array}
\right)
,
H_{2}=\left(
  \begin{array}{cccccc}
    1 & \ldots  & 1 & 1 & 0 & 0\\
 a_{1} & \ldots  & a_{q-1} & 0 & 1 & 0\\
    a_{1}^{2} & \ldots  & a_{q-1}^{2} & 0 & 0 & 1\\
\end{array}
\right)
,
G_{1}H_{2}^{T}=\left(
  \begin{array}{ccc}
    0 &  0 & 0  \\
    0 &  1  & 0 \\
    0 & 0  & 1 \\
\end{array}
\right).
$$
By Lemma 2, $\dim( C_{1} \cap C_{2})= \dim( C_{1})-\textnormal{rank}(G_{1}H_{2}^{T})=3-2=1$.

(iii) $l=2$:  Let $C_{1}$ be a $[q+2,3,q]_{q}$ MDS code with generator matrix $G_{1}$ and $C_{2}$ be a $[q+2,q-1,4]_{q}$ MDS code with parity check matrix  $H_{2}$, where
$$G_{1}=\left(
  \begin{array}{cccccc}
    1 &  \ldots & 1 &  1 &  0 &  0 \\
 a_{1} & \ldots  & a_{q-1} & 0 & 1 & 0\\
    a_{1}^{2} & \ldots  & a_{q-1}^{2} & 0 & 0 & 1\\
\end{array}
\right)
,
H_{2}=\left(
  \begin{array}{cccccc}
    1 & \ldots  & 1 & 1 & 0 & 0\\
 a_{1}^{-1} & \ldots  & a_{q-1}^{-1} & 0 & 1 & 0\\
    a_{1} & \ldots  & a_{q-1} & 0 & 0 & 1\\
\end{array}
\right)
,
G_{1}H_{2}^{T}=\left(
  \begin{array}{ccc}
    0 &  0 & 0  \\
    0 &  0  & 0 \\
    0 & 0  & 1 \\
\end{array}
\right).
$$
By Lemma 2, $\dim( C_{1} \cap C_{2})= \dim( C_{1})-\textnormal{rank}(G_{1}H_{2}^{T})=3-1=2$.

(iv) $l=3$:  Let $C_{1}$ be a $[q+2,3,q]_{q}$ MDS code with generator matrix $G_{1}$ and $C_{2}$ be a $[q+2,q-1,4]_{q}$ MDS code with parity check matrix  $H_{2}$, where
$$G_{1}=\left(
  \begin{array}{cccccc}
    1 &  \ldots & 1 &  1 &  0 &  0 \\
 a_{1} & \ldots  & a_{q-1} & 0 & 1 & 0\\
    a_{1}^{2} & \ldots  & a_{q-1}^{2} & 0 & 0 & 1\\
\end{array}
\right)
,
H_{2}=\left(
  \begin{array}{cccccc}
    1 & \ldots  & 1 & 1 & 0 & 0\\
 a_{1}^{-1} & \ldots  & a_{q-1}^{-1} & 0 & 1 & 0\\
    a_{1}^{-2} & \ldots  & a_{q-1}^{-2} & 0 & 0 & 1\\
\end{array}
\right)
,
G_{1}H_{2}^{T}=\left(
  \begin{array}{ccc}
    0 &  0 & 0  \\
    0 &  0  & 0 \\
    0 & 0  & 0 \\
\end{array}
\right).
$$
By Lemma 2, $\dim( C_{1} \cap C_{2})= \dim( C_{1})-\textnormal{rank}(G_{1}H_{2}^{T})=3-0=3$. \qed
\end{proof}

In summary, by Proposition 1 \cite[Proposition 3.1]{14} and Theorems 3,4, all possible linear $l$-intersection pairs of MDS codes over $\mathbb{F}_{q}$ with length $n\leq q+1$ are given. By Theorems 5,6 and 7, we give all possible linear $l$-intersection pairs of MDS codes over $\mathbb{F}_{2^{m}}$ with length $n=2^{m}+2\geq6$. As a result, all possible linear $l$-intersection pairs of MDS codes are given as follows.
\begin{theorem}
Let $q\geq3$ be a prime power and $n,k_{1}, k_{2}, l$ be non-negative integers. There exists a linear $l$-intersection
pair of MDS codes with parameters $[n, k_{1},n-k_{1}+1]_{q}$ and $[n, k_{2}, n-k_{2}+1]_{q}$ if one of the following conditions holds:\\
(i) $n\leq q+1$, $k_{1},k_{2}\leq n-1$, $\max \{k_{1}+k_{2}-n,0\} \leq l \leq \min\{k_{1}, k_{2}\}$ $($except $(n,k_{1},k_{2},l)\in \{(q+1,2,1,1),(q+1,1,2,1) \}$$)$; \\
(ii) $q=2^{m}\geq 4$, $n=q+2$, $(k_{1},k_{2})\in \{(3,q-1), (q-1,3),(3,3)\}$, $ 0 \leq l\leq 3$;\\
(iii) $q=2^{m}\geq 4$, $n=q+2$, $(k_{1},k_{2})=(q-1,q-1)$, $q-4 \leq l\leq q-1$.
\end{theorem}
\section{Constructions of pure MDS AEAQECCs}
In this section, we utilize Theorem 8 to give a complete characterization of pure MDS AEAQECCs. First, we give a useful lemma as follows.
\begin{lemma}
Let $C_{1}$ and $C_{2}$ be two MDS codes. If $C_{1}\nsubseteq C_{2}$, then $$wt(C_{1}\setminus (C_{1}\cap C_{2}))=wt(C_{1}).$$
\end{lemma}
\begin{proof}
Let $C_{1}$ be an $[n,k_{1},n-k_{1}+1]_{q}$ MDS code and $\mathcal{A}=\{\bm{a} \in \mathbb{F}_{2}^{n}: wt(\bm{a})=n-k_{1}+1\}$, then $|\mathcal{A}|={ n \choose n-k_{1}+1}$. For any $\bm{a}\in \mathcal{A}$, we define that $C_{\bm{a}}=\{\bm{c} \in C_{1}: wt(\bm{c})=n-k_{1}+1, wt(\bm{c},\bm{a})=n-k_{1}+1 \}$, where $wt(\bm{c},\bm{a})=\sharp \{i:(c_{i},a_{i})\neq (0,0)\}$. For any $\bm{a},\bm{b}\in \mathcal{A}$ and $\bm{a}\neq \bm{b}$, we can easily find that $|C_{\bm{a}}| = q-1$ and $C_{\bm{a}} \cap C_{\bm{b}}=\emptyset$, thus $|\bigcup_{\bm{a}\in \mathcal{A}}C_{\bm{a}}|=(q-1){ n \choose n-k_{1}+1}$. By Theorem 1, the number of $\bm{c}\in C_{1}$ with weight $n-k_{1}+1$ is $(q-1){ n \choose n-k_{1}+1}$, then $\bigcup_{\bm{a}\in \mathcal{A}}C_{\bm{a}}=\{\bm{c} \in C_{1}: wt(\bm{c})=n-k_{1}+1\}$.

Choose $\bm{a}_{i}=(\underbrace{0,\ldots,0}_{i-1},\underbrace{1,\ldots,1}_{n-k_{1}+1},0,\ldots,0) \in \mathcal{A}$ for $i=1,\ldots,k_{1}$ and $\bm{c}_{i}\in C_{\bm{a}_{i}}$. Obviously, $\bm{c}_{1},\ldots,\bm{c}_{k_{1}}$ are linearly independent, hence, $\{\bm{c}_{i}\}_{i=1,\ldots,k_{1}}$ is a basis of $C_{1}$. If  all $\bm{c}_{i} \in C_{2}$ for $i=1,\ldots,k_{1}$, then $C_{1} \subseteq C_{2}$, which leads to a contradiction. Therefore, there exists a $\bm{c}\in \{\bm{c}_{i}\}_{i=1,\ldots,k_{1}}$ with $wt(\bm{c})=n-k_{1}+1$ satisfying $\bm{c} \notin C_{2}$, i.e., $\bm{c} \in C_{1}\setminus (C_{1}\cap C_{2})$, then $$ n-k_{1}+1=wt(C_{1})\leq wt(C_{1}\setminus (C_{1}\cap C_{2}))\leq wt(\bm{c})=n-k_{1}+1.$$
Therefore, the lemma holds.\qed
\end{proof}
\begin{theorem}
Let $q\geq3$ be a prime power, $n,k_{1}, k_{2}, l$ be non-negative integers. There exists a pure MDS  $[[n,k_{2}-l,(k_{1}+1)/(n-k_{2}+1),k_{1}-l]]_{q}$ AEAQECC  if one of the following conditions holds:\\
(i) $n\leq q+1$, $k_{1},k_{2}\leq n-1$, $\max \{k_{1}+k_{2}-n,0\} \leq l <\min\{k_{1}, k_{2}\}$; \\
(ii) $q=2^{m}\geq 4$, $n=q+2$, $(k_{1},k_{2})\in \{(3,q-1), (q-1,3),(3,3)\}$, $0\leq l\leq 2$;\\
(iii) $q=2^{m}\geq 4$, $n=q+2$, $(k_{1},k_{2})=(q-1,q-1)$, $q-4\leq l\leq q-2$.
\end{theorem}
\begin{proof}
Let $C_{1}$ and $C_{2}^{\bot}$ be MDS codes with parameters $[n, k_{1},n-k_{1}+1]_{q}$ and $[n, k_{2}, n-k_{2}+1]_{q}$ respectively, where $n,k_{1},k_{2},l$ satisfy one of the conditions in Theorem 8. Then they form a linear $l$-intersection pair by Theorem 8, i.e., $\dim(C_{1}\cap C_{2}^{\bot})=l$.

By Theorem 2, $C_{1}$ and the dual code $C_{2}$ of $C_{2}^{\bot}$ can be used to construct an AEAQECC  with parameters $[[n,n-k_{1}-(n-k_{2})+c,d_{z}/d_{x},c]]_{q}$, where $d_{z}=wt\big(C_{1}^{\bot}\setminus (C_{2}\cap C_{1}^{\bot})\big)$,
$d_{x}=wt\big(C_{2}^{\bot}\setminus (C_{1}\cap C_{2}^{\bot})\big)$ and $c=\textnormal{rank}(G_{1}G_{2}^{T})=\dim(C_{1})-\dim(C_{1}\cap C_{2}^{\bot})=k_{1}-l$. Hence, the parameters are $[[n,k_{2}-l,d_{z}/d_{x},k_{1}-l]]_{q}$.

When $l=\min\{k_{1}, k_{2}\}$, we have $C_{1} \cap C_{2}^{\bot}\in\{ C_{1}, C_{2}^{\bot} \}$. If $C_{1} \cap C_{2}^{\bot}=C_{1}$, then $c=k_{1}-k_{1}=0$, i.e., the code doesn't need entanglement. If $C_{1} \cap C_{2}^{\bot}=C_{2}^{\bot}$, then $wt\big(C_{2}^{\bot}\setminus (C_{1} \cap C_{2}^{\bot})\big)=wt(C_{2}^{\bot}\setminus C_{2}^{\bot})=0$, which  makes no sense. Therefore, we just consider the parameters $n,k_{1}, k_{2}, l$ satisfy one of the conditions in Theorem 8 and $ l < \min\{k_{1}, k_{2}\}$, i.e., the parameters should satisfy one of the (i),(ii) and (iii).

For $ l < \min\{k_{1}, k_{2}\}$, it follows that $ C_{2}^{\bot} \nsubseteq C_{1}$ and $ C_{1}^{\bot} \nsubseteq C_{2}$. Therefore, by Lemma 4, $d_{z}=wt\big(C_{1}^{\bot}\setminus (C_{2}\cap C_{1}^{\bot})\big)=wt(C_{1}^{\bot})=k_{1}+1$ and
$d_{x}=wt\big(C_{2}^{\bot}\setminus (C_{1}\cap C_{2}^{\bot})\big)=wt(C_{2}^{\bot})=n-k_{2}+1$. Note that $$d_{x}+d_{z}=k_{1}+n-k_{2}+2=n-(k_{2}-l)+(k_{1}-l)+2,$$ it follows that it's a pure  MDS AEAQECC.

In summary, there exist  pure MDS $[[n,k_{2}-l,(k_{1}+1)/(n-k_{2}+1),k_{1}-l]]_{q}$ AEAQECCs  when $n,k_{1},k_{2},l$ satisfy one of the (i),(ii) and (iii). \qed
\end{proof}
\begin{remark}
Let $Q$ be an $[[n,n-k_{1}-k_{2}+c,d_{z}/d_{x},c]]_{q}$ AEAQECC constructed by linear codes $C_{1}$ and $C_{2}$ with parameters $[n,k_{1}]_{q}$ and  $[n,k_{2}]_{q}$ respectively. $Q$ is a pure MDS AEAQECC if and only if  $d_{z}=wt(C_{1}^{\bot})$ and
$d_{x}=wt(C_{2}^{\bot})$ with $d_{x}+d_{z}=k_{1}+k_{2}+2$. Note that $wt(C_{1}^{\bot})\leq k_{1}+1$ and $wt(C_{2}^{\bot})\leq k_{2}+1$. Therefore, $Q$ is a pure MDS AEAQECC if and only if $C_{1}$ and $C_{2}$ are MDS codes such that $d_{z}=wt(C_{1}^{\bot})$ and
$d_{x}=wt(C_{2}^{\bot})$. Assuming the validity of MDS Conjecture, the length of MDS codes over $\mathbb{F}_{q}$ is no more than $q+2$, hence, the length of pure MDS AEAQECCs  over $\mathbb{F}_{q}$ is no more than $q+2$ too. As a result, we obtain  pure MDS AEAQECCs for all possible parameters by Theorem 9.
\end{remark}
\section{Conclusions}
In this paper,  we firstly construct linear $l$-intersection pairs of MDS codes with parameters $[n, k_{1}, n-k_{1}+1]_{q}$ and $[n, k_{2}, n-k_{2}+1]_{q}$ where  $(n,k_{1},k_{2},l)=(q,l+1,l+1,l)$ for $0\leq l\leq q-2$ and $n=q+1$ with $1 \in \{ l, k_{1}-l, k_{2}-l \}$ for $k_{1},k_{2}\leq q$, which complement the results in \cite{14}. Moreover, we also construct all possible linear $l$-intersection pairs of MDS codes over $\mathbb{F}_{2^{m}}$ with length $n=2^{m}+2\geq6$. In summary, all possible linear $l$-intersection pairs of MDS codes are given. As an application, we utilize  linear $l$-intersection pairs of MDS codes to determine the required number of  maximally entangled states of an AEAQECC. As a result, a complete characterization of pure MDS AEAQECCs for all possible parameters is given.
\begin{acknowledgements}
The research of Z. Huang and F.-W. Fu  is supported in part by the National Key Research and Development Program of China (Grant No. 2018YFA0704703), the National Natural Science Foundation of China (Grant No. 61971243), the Natural Science Foundation of Tianjin (20JCZDJC00610), the Fundamental Research Funds for the Central Universities of China (Nankai University). The research of W. Fang is supported in part by the China  Postdoctoral  Science  Foundation  under  Grant  2020M670330, Guangdong Basic and Applied Basic Research Foundation under Grant 2019A1515110904.
\end{acknowledgements}

%
%



\end{document}